\documentclass[envcountsame,envcountsect, oribibl]{llncs}
\usepackage{amsmath,amssymb}
\usepackage{bbold, pgf, tikz, verbatim,stmaryrd,eurosym}
\usetikzlibrary{arrows,automata}

\DeclareMathOperator{\val}{val}

\DeclareMathOperator{\dom}{dom}
\DeclareMathOperator{\wt}{wt}

\DeclareMathOperator{\Run}{Run}

\DeclareMathOperator{\Label}{label}

\DeclareMathOperator{\dd}{d}

\newcommand{\Def}{\text{\sc Def}}

\newcommand{\Rec}{\text{\sc Rec}}

\newcommand{\wRdl}{\text{\sc wRdl}}
\newcommand{\Rdl}{\text{\sc Rdl}}

\newcommand{\Unamb}{\mathcal N^{\text{\sc Unamb}}}
\newcommand{\Det}{\mathcal N^{\text{\sc Det}}}
\newcommand{\Seq}{\mathcal N^{\text{\sc Seq}}}
\newcommand{\Unambb}{\mathcal H^{\text{\sc Unamb}}}
\newcommand{\Dett}{\mathcal H^{\text{\sc Det}}}
\newcommand{\Seqq}{\mathcal H^{\text{\sc Seq}}}

\newcommand{\nSeq}{\mathcal N}

\newcommand{\M}{\mathbb M}
\newcommand{\A}{\mathcal A}

\newcommand{\V}{\mathcal V}

\newcommand{\zero}{\mathbb 0}
\newcommand{\one}{\mathbb 1}

\newcommand{\TRUE}{\text{\sc True}}

\newcommand{\Rp}{\mathbb R_{\ge 0}}

\begin{document}
\title{A Nivat Theorem for Weighted Timed Automata and  Weighted Relative Distance Logic\thanks{The final version appeared in the Proceedings of the 41st International Colloquium on Automata, Languages, and Programming (ICALP 2014) and is available at link.springer.com; DOI: 10.1007/978-3-662-43951-7\_15}}
\author{Manfred Droste and Vitaly Perevoshchikov\thanks{Supported by  DFG Graduiertenkolleg 1763 (QuantLA)}}
\institute{Universit\"at Leipzig,  Institut f\"ur Informatik, \\ 
04109 Leipzig, Germany\\
\email{\{droste,perev\}@informatik.uni-leipzig.de}
}

\maketitle

\begin{abstract}
Weighted timed automata (WTA) model quantitative aspects of real-time systems like continuous consumption of memory, power or financial resources. They accept quantitative timed languages where every timed word is mapped to a value, e.g., a real number. In this paper, we prove a Nivat theorem for WTA which states that recognizable quantitative timed languages are exactly those which can be obtained from recognizable boolean timed languages with the help of several simple operations. We also introduce a weighted extension of relative distance logic developed by Wilke, and we show that our weighted relative distance logic and WTA are equally expressive. The proof of this result can be derived from our Nivat theorem and Wilke's theorem for relative distance logic. Since the proof of our Nivat theorem is constructive, the translation process from logic to automata and vice versa is also constructive. This leads to decidability results for weighted relative distance logic. 
\end{abstract}

\section{Introduction}

Timed automata introduced by Alur and Dill \cite{AD94} are a prominent model for real-time systems. 
Timed automata form finite representations of infinite-state automata for which various fundamental results from the theory of finite-state automata can be transferred to the timed setting. 
Although time has a quantitative nature, the questions asked in the theory of timed automata are of a qualitative kind. On the other side, quantitative aspects of systems, e.g., costs, probabilities and energy consumption can be modelled using weighted automata, i.e., classical nondeterministic automata with a transition weight function. The behaviors of weighted automata can be considered as quantitative languages (also known as formal power series) where every word carries a value. Semiring-weighted automata have been extensively studied in the literature (cf. \cite{BR88,Eil74,KS86} and the handbook of weighted automata \cite{DKV09}). 

Weighted extensions of timed automata are of much interest for the real-time community, since weighted timed automata (WTA) can model continuous time-dependent consumption of resources. In the literature, various models of WTA were considered, e.g., linearly priced timed automata \cite{ATP01,BFHL01, LBBF01},
multi-weighted timed automata with knapsack-problem objective \cite{LR05}, and WTA with measures like average, reward-cost ratio \cite{BBL04, BBL08} and discounting \cite{AT11, FL09, FL092}. In \cite{Qua10, Qua11}, WTA over semi\-rings were studied with respect to the classical automata-theoretic questions. However, various models, e.g., WTA with average and discounting measures as well as multi-weighted automata cannot be defined using semirings. For the latter situations, only several algorithmic problems were handled. But many questions whether the results known from the theories of timed and weighted automata also hold for WTA remain open. Moreover, there is no unified framework for WTA. 

The main goal of this paper is to build a bridge between the theories of WTA and timed automata. First, we develop a general model of {\em timed valuation monoids} for WTA. Recall that Nivat's theorem \cite{Ni68} is one of the fundamental characterizations of rational transductions and establishes a connection between rational transductions and rational languages. Our first main result is an extension of Nivat's theorem to WTA over timed valuation monoids. By Nivat's theorem for semiring-weighted automata described recently in \cite{DK}, recognizable quantitative languages are exactly those which can be constructed from recognizable languages using operations like morphisms and intersections. The proof of this result requires the fact that finite automata are determinizable. However, timed automata do not enjoy this property. Nevertheless, for idempotent timed valuation monoids which model all mentioned examples of WTA, we do not need determinization. In this case, our Nivat theorem for WTA is similar to the one for weighted automata. In the non-idempotent case, we give an example showing that this statement does not hold true. But in this case we can establish a connection between recognizable quantitative timed languages and sequentially, deterministically or unambiguously recognizable timed languages. 

As an application of our Nivat theorem, we provide a characterization of recognizable quantitative timed languages by means of quantitative logics. The classical B\"uchi-Elgot theorem \cite{Buc60} was extended to both weighted \cite{DG07, DG09, DM12} and timed settings \cite{Wil94,Wil942}. In \cite{Qua10, Qua11}, a semiring-weighted extension of Wilke's relative distance logic \cite{Wil94,Wil942} was considered. Here, we develop a different weighted version of relative distance logic based on our notion of timed valuation monoids. In our second main result, we show that this logic and WTA have the same expressive power. For the proof of this result, we use a new proof technique and our Nivat theorem to derive our result from the corresponding result for unweighted logic \cite{Wil94, Wil942}. Since the proof of our Nivat theorem is constructive, the translation process from weighted relative distance logic to WTA and vice versa is constructive. This leads to decidability results for weighted relative distance logic. In particular, based on the results of \cite{ATP01,BFHL01, LBBF01}, we show the decidability of several weighted extensions of the satisfiability problem for our logic. 

\section{Timed Automata}

An {\em alphabet} is a non-empty finite set.
Let $\Sigma$ be a non-empty set. A {\em finite word} over $\Sigma$ is a finite sequence $a_1 ... a_n$ where $n \ge 0$ and $a_1, ..., a_n \in \Sigma$. If $n \ge 1$, then we say that $w$ is {\em non-empty}. Let $\Sigma^+$ denote the set of all non-empty words over $\Sigma$. 
Let $\Rp$ denote the set of all non-negative real numbers. A {\em finite timed word} over $\Sigma$ is a finite word over $\Sigma \times \Rp$, i.e., a finite sequence $(a_1, t_1) ... (a_n, t_n)$ where $n \ge 0$, $a_1, ..., a_n \in \Sigma$ and $t_1, ..., t_n \in \Rp$. Let $|w| = n$ and $\langle w \rangle = t_1 + ... + t_n$ and let $\mathbb T \Sigma^+ = (\Sigma \times \Rp)^+$,
the set of all non-empty finite timed words. Any set $\mathcal L \subseteq \mathbb T \Sigma^+$ of timed words is called a {\em timed language}. 

Let $C$ be a finite set of {\em clock variables} ranging over $\Rp$. A {\em clock constraint} over $C$ is either $\TRUE$ or (if $C$ is non-empty) a conjunction of formulas of the form $x \bowtie c$ where $x \in C$, $c \in \mathbb N$ and ${\bowtie} \in \{<, \le, =, \ge, >\}$. Let $\Phi(C)$ denote the set of all clock constraints over $C$. A {\em clock valuation} over $C$ is a mapping $\nu: C \to \Rp$ which assigns a value to each clock variable. Let $\Rp^C$ be the set of all clock valuations over $C$. The {\em satisfaction relation} ${\models} \subseteq \Rp^C \times \Phi(C)$ is defined as usual. Now let $\nu \in \Rp^C$, $t \in \Rp$ and $\Lambda \subseteq C$. Let $\nu + t$ denote the clock valuation $\nu' \in \Rp^C$ such that $\nu'(x) = \nu(x) + t$ for all $x \in C$. Let $\nu[\Lambda := 0]$ denote the clock valuation $\nu' \in \Rp^C$ such that $\nu'(x) = 0$ for all $x \in \Lambda$ and $\nu'(x) = \nu(x)$ for all $x \notin \Lambda$.

\begin{definition}
Let $\Sigma$ be an alphabet. A {\em timed automaton} over $\Sigma$ is a tuple ${\mathcal A = (L, C, I, E, F)}$ such that $L$ is a finite set of {\em locations}, $C$ is a finite set of {\em clocks}, $I, F \subseteq L$ are sets of {\em initial} resp. {\em final} locations and ${E \subseteq L \times \Sigma \times \Phi(C) \times 2^C \times L}$ is a finite set of {\em edges}.
\end{definition}
For an edge $e = (\ell, a, \phi, \Lambda, \ell')$, let $\Label(e) = a$ be the {\em label} of $e$. A {\em run} of $\mathcal A$ is a finite sequence
\begin{equation}
\label{Eq:DefRun}
\rho = (\ell_0, \nu_0) \xrightarrow{t_1} \xrightarrow{e_1} (\ell_1, \nu_1) \xrightarrow{t_2} \xrightarrow{e_2} ... \xrightarrow{t_n} \xrightarrow{e_n} (\ell_n, \nu_n)
\end{equation}
where $n \ge 1$, $\ell_0, \ell_1, ..., \ell_n \in L$, $\nu_0, \nu_1, ..., \nu_n \in \Rp^C$, $t_1, ..., t_n \in \Rp$ and ${e_1, ..., e_n \in E}$ satisfy the following conditions: $\ell_0 \in I$, $\nu_0(x) = 0$ for all $x \in C$, $\ell_n \in F$ and, for all $1 \le i \le n$, $e_i = (\ell_{i-1}, a_i, \phi_i, \Lambda_i, \ell_i)$ for some $a_i \in \Sigma$, $\phi_i \in \Phi(C)$ and $\Lambda_i \subseteq C$ such that $\nu_{i-1} + t_i \models \phi_i$ and $\nu_i = (\nu_{i-1} + t_i)[\Lambda_i := 0]$. The {\em label} of $\rho$ is the timed word $\Label(\rho) = (\Label(e_1), t_1) ... (\Label(e_n), t_n) \in \mathbb T \Sigma^+$. For any timed word $w \in \mathbb T \Sigma^+$, let $\Run_{\A}(w)$ denote the set of all runs $\rho$ of $\mathcal A$ such that $\Label(\rho) = w$. Let $\mathcal L(\A) = \{w \in \mathbb T \Sigma^+ \; | \; \Run_{\A}(w) \neq \emptyset\}$. We say that an arbitrary timed language $\mathcal L \subseteq \mathbb T \Sigma^+$ is {\em recognizable} if there exists a timed automaton over $\Sigma$ such that $\mathcal L(\A) = \mathcal L$. Let ${\mathcal A = (L, C, I, E, F)}$ be a timed automaton over $\Sigma$. We say that $\mathcal A$ is {\em unambiguous} if $|\Run_{\A}(w)| \le 1$ for all $w \in \mathbb T \Sigma^+$. We call $\A$ {\em deterministic} if $|I| = 1$ and, for all $e_1 = (\ell, a, \phi_1, \Lambda_1, \ell_1) \in E$ and $e_2 = (\ell, a, \phi_2, \Lambda_2, \ell_2) \in E$ with $e_1 \neq e_2$, there exists no clock valuation $\nu \in \Rp^C$ with $\nu \models \phi_1 \wedge \phi_2$.
We call $\A$ {\em sequential} if $|I| = 1$ and, for all $e_1 = (\ell, a, \phi_1, \Lambda_1, \ell_1) \in E$ and $e_2 = (\ell, a, \phi_2, \Lambda_2, \ell_2) \in E$, we have $e_1 = e_2$; this property can be viewed as a strong form of determinism. Based on these notions, we can define {\em sequentially recognizable}, {\em deterministically recognizable} and {\em unambiguously recognizable} timed languages. 

\section{Weighted Timed Automata}

In this section, we introduce a general model of weighted timed automata (WTA) over {\em timed valuation monoids}. We will show that our new model covers a variety of situations known from the literature: 
linearly priced timed automata \cite{ATP01, BFHL01,LBBF01} and WTA with the measures like average \cite{BBL04, BBL08} and discounting \cite{AT11,FL09,FL092}.

A {\em timed valuation monoid} is a tuple $\M = (M, +, \val, \zero)$ where $(M, +, \zero)$ is a commutative monoid and $\val: \mathbb T (M \times M)^+ \to M$ is a {\em timed valuation function}.
We will say that $M$ is the {\em domain} of $\M$.
We say that $\M$ is {\em idempotent} if $+$ is idempotent, i.e., $m + m = m$ for all $m \in M$.

Let $\Sigma$ be an alphabet and $\M = (M, +, \val, \zero)$ a timed valuation monoid. A {\em weighted timed automaton} (WTA) over $\Sigma$ and $\M$ is a tuple $\A = (L, C, I, E, F, \wt)$ where $(L, C, I, E, F)$ is a timed automaton over $\Sigma$ and $\wt: L \cup E \to M$ is a {\em weight function}. Let $\rho$ be a run of $\A$ of the form (\ref{Eq:DefRun}). Let $\wt^{\sharp}(\rho) \in \mathbb T (M \times M)^+$ be the timed word $(u_1, t_1) ... (u_n, t_n)$ where, for all $1 \le i \le n$, $u_i = (\wt(\ell_{i-1}), \wt(e_i))$. Then, the {\em weight} of $\rho$ is defined as $\wt_{\A}(\rho) = \val(\wt^{\sharp}(\rho)) \in M$. The {\em behavior} of $\mathcal A$ is the mapping $||\A||: \mathbb T \Sigma^+ \to M$ defined by
$
{||\A||(w) = \sum (\wt_{\A}(\rho) \; | \; \rho \in \Run_{\A}(w))}
$
 for all $w \in \mathbb T \Sigma^+$. A {\em quantitative timed language} (QTL) over $\M$ is a mapping ${\mathbb L: \mathbb T \Sigma^+ \to M}$. We say that $\mathbb L$ is {\em recognizable} if there exists a WTA $\A$ over $\Sigma$ and $\M$ such that $\mathbb L = ||\A||$.

\begin{example}
\label{Ex:TVM}
All of the subsequent WTA model the property that staying in a location invokes costs depending on the length of the stay; the subsequent transition also invokes costs but happens instantaneously.
We assume that, for all $x \in \mathbb R \cup \{\infty\}$, ${x \cdot \infty = \infty \cdot x = \infty}$ and $x + \infty = \infty + x = \infty$.
\begin{itemize}
\item [(a)] {\em Linearly priced timed automata} were considered in \cite{ATP01, BFHL01, LBBF01}. We can describe this model by the timed valuation monoid  \linebreak ${\M^{\text{sum}} = (\mathbb R \cup \{\infty\}, \min, \val^{\text{sum}}, \infty)}$ where $\val^{\text{sum}}$ is defined by
$\val^{\text{sum}}(v) = \sum_{i = 1}^n (m_i \cdot t_i + m'_i)$ for all 
$v = ((m_1, m'_1), t_1) ... ((m_n, m'_n), t_n) \in \mathbb T (M \times M)^+$.

\item [(b)] The situation of the average behavior for WTA considered in \linebreak \cite{BBL04, BBL08} can be described by means of the timed valuation monoid \linebreak ${\M^{\text{avg}} = (\mathbb R \cup \{\infty\}, \min, \val^{\text{avg}}, \infty)}$ where $\val^{\text{avg}}$ is defined as follows. Let ${v = ((m_1, m'_1), t_1) ... ((m_n, m'_n), t_n) \in \mathbb T (M \times M)^+}$. If $\langle v \rangle > 0$, then we let
$
\val^{\text{avg}}(v) = \frac{\sum_{i = 1}^n (m_i \cdot t_i + m_i')}{\sum_{i = 1}^n t_i}.
$
If $\langle v \rangle = 0$, $m_1 = ... = m_n \in \mathbb R$ and $m_1' = ... = m_n' = 0$, then we put $\val^{\text{avg}}(v) = m_1$. Otherwise, we put $\val^{\text{avg}}(v) = \infty$.

\item [(c)] The model of WTA with the discounting measure was investigated in \cite{AT11, FL09, FL092}. These WTA can be considered as WTA over the timed valuation monoid ${\M^{\text{disc}_{\lambda}} = (\mathbb R \cup \{\infty\}, \min, \val^{\text{disc}_{\lambda}}, \infty)}$ where $0 < \lambda < 1$ is a {\em discounting factor} and $\val^{\text{disc}_{\lambda}}$ is defined for all $v = ((m_1, m_1'), t_1) ... ((m_n, m_n'), t_n) \in \mathbb T (M \times M)^+$ by 
$
\val^{\text{disc}_\lambda}(v) = \sum_{i = 1}^n \lambda^{t_1 + ... + t_{i-1}} \cdot \big(\int_{0}^{t_i} m_i \cdot \lambda^{\tau} d \tau + \lambda^{t_i} \cdot m_i'\big).
$
\end{itemize}
Note that the timed valuation monoids $\M^{\text{sum}}$, $\M^{\text{avg}}$ and $\M^{\text{disc}_{\lambda}}$ are idempotent. 
\end{example}

\section{Closure Properties}
\label{Sect:Closure}

In this section, we consider several closure properties of recognizable quantitative timed languages which we will use for the proof of our Nivat theorem and which could be of independent interest. For lack of space, we will omit the proofs.

Let $\Sigma$ be a set, $\Gamma$ an alphabet and $h: \Gamma \to \Sigma$ a mapping. For a timed word ${v = (\gamma_1, t_1) ... (\gamma_n, t_n) \in \mathbb T \Gamma^+}$, we let $h(v) = (h(\gamma_1), t_1) ... (h(\gamma_n), t_n) \in \mathbb T \Sigma^+$. Then, for a QTL $r: \mathbb T \Gamma^+ \to M$ over $\M$, we define the QTL $h(r): \mathbb T \Sigma^+ \to M$ over $\M$ by
$
h(r)(w) = \sum (r(v) \; | \; v \in \mathbb T \Gamma^+ \text{ and } h(v) = w)
$
for all $w \in \mathbb T \Sigma^+$. Observe that for any $w \in \mathbb T \Sigma^+$ there are only finitely many $v \in \mathbb T \Gamma^+$ with $h(v) = w$, hence the sum exists in $(M, +)$.

\begin{lemma}
\label{Lemma:Homo}
Let $\Sigma, \Gamma$ be alphabets, $\M = (M, +, \val, \zero)$ a timed valuation monoid and $h: \Gamma \to \Sigma$ a mapping. If $r: \mathbb T \Gamma^+ \to M$ is a recognizable QTL over $\M$, then the QTL $h(r)$ is also recognizable.
\end{lemma}
For the proof of this lemma, we use a similar construction as in \cite{DV10}, Lemma 1.

Let $g: \Sigma \to M \times M$ be a mapping. We denote by ${\val \circ g: \mathbb T \Sigma^+ \to M}$ the QTL over $\M$ defined for all $w \in \mathbb T \Sigma^+$ by ${(\val \circ g)(w) = \val(g(w))}$. We say that a timed valuation monoid $\M = (M, +, \val, \zero)$ is {\em location-independent} if, for any $v = ((m_1, m_1'), t_1) ... ((m_n, m_n'), t_n) \in \mathbb T (M \times M)^+$ and ${v' = ((k_1, k_1'), t_1) ... ((k_n, k_n'), t_n) \in \mathbb T (M \times M)^+}$ with $m_i' = k_i'$ for all $1 \le i \le n$, we have $\val(v) = \val(v')$.

\begin{lemma}
\label{Lemma:Comp}
Let $\Sigma$ be an alphabet, $\M = (M, +, \val, \zero)$ a timed valuation monoid and $g: \Sigma \to M \times M$ a mapping. Then, $\val \circ g$ is unambiguously recognizable. If $\M$ is location-independent, then $\val \circ g$ is sequentially recognizable.
\end{lemma}

However, in general, $\val \circ g$ is not deterministically recognizable (and hence not sequentially recognizable). Let $\Sigma = \{a, b\}$ and $\M = \M^{\text{sum}}$ as in Example \ref{Ex:TVM} (a). Let $g(a) = (1, 0)$ and $g(b) = (2, 0)$. Then, one can show that $\val \circ g$ is not deterministically recognizable.

Let $\mathcal L \subseteq \mathbb T \Sigma^+$ be a timed language and $r: \mathbb T \Sigma^+ \to M$ a QTL over $\M$. The {\em intersection} $(r \cap \mathcal L): \mathbb T \Sigma^+ \to M$ is the QTL over $\M$ defined by $(r \cap \mathcal L)(w) = r(w)$ if $w \in \mathcal L$ and $(r \cap \mathcal L)(w) = \zero$ if $w \in \mathbb T \Sigma^+ \setminus \mathcal L$. 

\begin{example}
\label{Ex:Bad}
As opposed to weighted untimed automata, recognizable quantitative timed  languages are not closed under the intersection with recognizable timed languages. Let $\Sigma$ be a singleton alphabet and $\mathcal L$ a recognizable timed language over $\Sigma$ which is not unambiguously recognizable. Wilke \cite{Wil94} showed that such a language exists. Consider the non-idempotent and location-independent timed valuation monoid $\M = (\mathbb N, +, \val, 0)$ where $+$ is the usual addition of natural numbers and $\val(v) = m_1' \cdot ... \cdot m_n'$ for all $v = ((m_1, m_1'), t_1) ... ((m_n, m_n'), t_n) \in \mathbb T (\mathbb N \times \mathbb N)^+$. Let the QTL ${r: \mathbb T \Sigma^+ \to \mathbb N}$ over $\M$ be defined by $r(w) = 1$ for all $w \in \mathbb T \Sigma^+$. Then, $r$ is recognizable but $r \cap \mathcal L$ is not recognizable. 
\end{example}

Nevertheless, the intersection enjoys the following closure properties. 

\begin{lemma}
\label{Lemma:Inter}
Let $\Sigma$ be an alphabet, $\M = (M, +, \val, \zero)$ a timed valuation monoid, $\mathcal L \subseteq \mathbb T \Sigma^+$ a recognizable timed language and $r: \mathbb T \Sigma^+ \to M$ a recognizable  QTL over $\M$. If $\M$ is idempotent, then $r \cap \mathcal L$ is recognizable. If $\mathcal L$ is unambiguously recognizable, then $r \cap \mathcal L$ is recognizable. If $\mathcal L, r$ are unambiguously (deterministically, sequentially) recognizable, then $r \cap \mathcal L $ is also unambiguously (deterministically, sequentially) recognizable. 
\end{lemma}
For the proof, we use a kind of product construction for timed automata. 

\section{A Nivat Theorem for Weighted Timed Automata}

Nivat's theorem \cite{Ni68} (see also \cite{Be69}, Theorem 4.1) is one of the fundamental characterizations of rational transductions and establishes a connection between rational transductions and rational languages. A version for semiring-weighted automata was given in  \cite{DK}; this shows a connection between recognizable quantitative and qualitative languages. In this chapter, we prove a Nivat-like theorem for recognizable quantitative timed languages.

Let $\Sigma$ be an alphabet and $\M = (M, +, \val, \zero)$ a timed valuation monoid. Let $\Rec(\Sigma, \M)$ denote the collection of all QTL recognizable by a WTA over $\Sigma$ and $\M$. Let $\nSeq(\Sigma, \M)$ (with $\mathcal N$ standing for Nivat) denote the set of all QTL ${\mathbb L: \mathbb T \Sigma^+ \to M}$ over $\M$ such that there exist an alphabet $\Gamma$, mappings ${h: \Gamma \to \Sigma}$ and ${g: \Gamma \to M \times M}$ and a recognizable timed language $\mathcal L \subseteq \mathbb T \Gamma^+$ such that $\mathbb L = h((\val \circ g) \cap \mathcal L)$. Let the collection $\Seq(\Sigma, \M)$ be defined like $\nSeq(\Sigma, \M)$ with the only difference that $\mathcal L$ is sequentially recognizable. The collections $\Unamb(\Sigma, \M)$ and $\Det(\Sigma, \M)$ are defined similarly using unambiguously resp. deterministically recognizable timed languages.

Our Nivat theorem for weighted timed automata is the following.

\begin{theorem}
\label{Theorem:Nivat}
Let $\Sigma$ be an alphabet and $\M$ a timed valuation monoid. Then,
$\Rec(\Sigma, \M) = \Seq(\Sigma, \M) = \Det(\Sigma, \M) =  \Unamb(\Sigma, \M) \subseteq \nSeq(\Sigma, \M)$. \linebreak
If $\M$ is idempotent,  then $\Rec(\Sigma, \M) = \nSeq(\Sigma, \M)$.
\end{theorem}

As opposed to the result of \cite{DK} for weighted untimed automata, the equality $\Rec(\Sigma, \M) = \nSeq(\Sigma, \M)$ does not always hold: let $\Sigma$, $\M$, $\mathcal L$ and $r$ be defined as in Example \ref{Ex:Bad}. Then, one can show that $r \cap \mathcal L \in \nSeq(\Sigma, \M) \setminus \Rec(\Sigma, \M)$. 

The proof of Theorem \ref{Theorem:Nivat} is based on the closure properties of WTA (cf. Sect. \ref{Sect:Closure}) and the following lemma. 

\begin{lemma}
\label{Lemma:Transitions}
Let $\Sigma$ be an alphabet and $\M$ a timed valuation monoid. Then, 
$\Rec(\Sigma, \M) \subseteq \Seq(\Sigma, \M)$.
\end{lemma}

\begin{proof}[Sketch]
Let $\A = (L, C, I, E, F, \wt)$ be a WTA over $\Sigma$ and $\M$. Let $\Gamma = E$. We define the mappings $h: \Gamma \to \Sigma$ and ${g: \Gamma \to M \times M}$ for all $\gamma = (\ell, a, \phi, \Lambda, \ell') \in \Gamma$ by $h(\gamma) = a$ and $g(\gamma) = (\wt(\ell), \wt(\gamma))$. Let $\mathcal L$ be the set of all timed words $w = (\gamma_1, \tau_1) ... (\gamma_n, \tau_n)$ such that there exists a run $\rho$ of $\A$ of the form (\ref{Eq:DefRun}) with $\gamma_i = e_i$ and $\tau_i = t_i$ for all $1 \le i \le n$. It can be shown that $\mathcal L$ is sequentially recognizable and $||\A|| = h((\val \circ g) \cap \mathcal L) \in \Seq(\Sigma, \M)$. \qed
\end{proof}

 Let $\Sigma$ be an alphabet and $\M$ a timed valuation monoid with the domain $M$. Let $\Unambb(\Sigma, \M)$ denote the collection of all QTL ${\mathbb L: \mathbb T \Sigma^+ \to M}$ over $\M$ such that there exist an alphabet $\Gamma$, a mapping $h: \Gamma \to \Sigma$ and an unambiguously recognizable QTL $r: \mathbb T \Gamma^+ \to M$ over $\M$ such that $\mathbb L = h(r)$. 
The collections $\Seqq(\Sigma, \M)$ and $\Dett(\Sigma, \M)$ are defined like $\Unambb(\Sigma, \M)$ with the only difference that $r$ is sequentially resp. deterministically recognizable.

As a corollary from Theorem 5.1, we establish the following connections between recognizable and unambiguously, sequentially and deterministically recognizable QTL. For the proof of this corollary, we apply Theorem \ref{Theorem:Nivat} and closure properties of WTA considered in Sect. \ref{Sect:Closure}. 
\begin{corollary} 
\label{Cor:Unamb}
Let $\Sigma$ be an alphabet and $\M$ a timed valuation monoid. Then, $\Seqq(\Sigma, \M) = \Dett(\Sigma, \M) \subseteq \Unambb(\Sigma, \M) = \Rec(\Sigma, \M)$. If $\M$ is location-independent, then $\Seqq(\Sigma, \M) = \Rec(\Sigma, \M)$.
\end{corollary}

However, the equality $\Seqq(\Sigma, \M) = \Rec(\Sigma, \M)$ does not always hold. Let ${\Sigma = \{a, b\}}$ and $\M = \M^{\text{sum}}$ be the timed valuation monoid as in Example \ref{Ex:TVM} (a); note that $\M$ is not location-independent. Consider the QTL ${\mathbb L: \mathbb T \Sigma^+ \to M}$ over $\M$ defined for all $w = (a_1, t_1) ... (a_n, t_n)$ by $\mathbb L(w) = t_1$ if $a_1 = a$ and $\mathbb L(w) = 2 \cdot t_1$ otherwise. We can show that $\mathbb L \in \Rec(\Sigma, \M) \setminus \Seqq(\Sigma, \M)$. 

\section{Weighted Relative Distance Logic}

In this section, we develop a weighted relative distance logic. Relative distance logic on timed words was introduced by Wilke in \cite{Wil94, Wil942}. It was shown that restricted relative distance logic and timed automata have the same expressive power. Here, we will derive a weighted version of this result. We will show that the proof of our result can be deduced from Wilke's result and our Nivat theorem for WTA.

We fix a countable set $V_1$ of {\em first-order variables} and a countable set $V_2$  of {\em second-order variables} such that $V_1 \cap V_2 = \emptyset$. Let $V = V_1 \cup V_2$.

\subsection{Relative Distance Logic}

Let $\Sigma$ be an alphabet. The set $\Rdl(\Sigma)$ of {\em relative distance formulas} over $\Sigma$ is defined by the grammar:
$$
\varphi \; ::= \; P_a(x) \; | \; x \le y \; | \; X(x) \; | \; \dd_{\leftarrow}^{{\bowtie} c}(X, x) \; | \;  \lnot \varphi \; | \; \varphi \vee \varphi \; | \; \exists x. \varphi \; | \; \exists X. \varphi
$$
where $a \in \Sigma$, $x, y \in V_1$, $X \in V_2$, ${\bowtie} \in \{<, \le, =, \ge, >\}$ and $c \in \mathbb N$.
The formulas of the form $\dd_{\leftarrow}^{{\bowtie c}}(X, x)$ are called {\em past formulas}. 

Let $w = (a_1, t_1) ... (a_n, t_n) \in \mathbb T \Sigma^+$ be a timed word. For every $1 \le i \le n$, let $\langle w \rangle_i = t_1 + ... + t_i$. The {\em domain} of $w$ is the set $\dom(w) = \{1, ..., n\}$ of {\em positions} of $w$. Let $y \in \dom(w)$, $Y \subseteq \dom(w)$, ${\bowtie} \in \{<, \le, =, \ge, >\}$ and $c \in \mathbb N$. Then, we write $\dd^{{\bowtie} c, w}_{\leftarrow}(Y, y)$ iff either there exists a position $z \in Y$ such that $z < y$ and, for the greatest such position $z$, $\langle w \rangle_y - \langle w \rangle_z \bowtie c$, or there exists no position $z \in Y$ with $z < y$, and $\langle w \rangle_y \bowtie c$.
A {\em $w$-assignment} is a mapping $\sigma: V \to \dom(w) \cup 2^{\dom(w)}$ such that $\sigma(V_1) \subseteq \dom(w)$ and $\sigma(V_2) \subseteq 2^{\dom(w)}$. We define the {\em update} $\sigma[x/i]$ to be the $w$-assignment such that $\sigma[x/i](x) = i$ and $\sigma[x/i](y) = \sigma(y)$ for all $y \in V \setminus \{x\}$. Similarly, for $X \in V_2$ and $I \subseteq \dom(w)$, we define the update $\sigma[X/I]$. Let $\varphi \in \Rdl(\Sigma)$ and $\sigma$ be a $w$-assignment. The definition that the pair $(w, \sigma)$ {\em satisfies} the formula $\varphi$, written $(w, \sigma) \models \varphi$, is given inductively on the structure of $\varphi$ as usual for MSO logic where, for the new formulas $\dd^{{\bowtie} c}_{\leftarrow}(X, x)$, we put $(w, \sigma) \models \dd^{{\bowtie} c}_{\leftarrow}(X, x)$ iff $\dd_{\leftarrow}^{{\bowtie} c, w}(\sigma(X), \sigma(x))$.

A formula $\varphi \in \Rdl(\Sigma)$ is called a {\em sentence} if every variable occurring in $\varphi$ is bound by a quantifier. Note that, for a sentence $\varphi \in \Rdl(\Sigma)$, the relation $(w, \sigma) \models \varphi$ does not depend on $\sigma$, i.e., for any $w$-assignments $\sigma_1, \sigma_2$, $(w, \sigma_1) \models \varphi$ iff ${(w, \sigma_2) \models \varphi}$. Then, we will write $w \models \varphi$. For a sentence $\varphi \in \Rdl(\Sigma)$, let ${\mathcal L(\varphi) = \{w \in \mathbb T \Sigma^+ \; | \; w \models \varphi\}}$, the timed language {\em defined} by $\varphi$. 
Let $\Delta \subseteq \Rdl(\Sigma)$. We say that a timed language $\mathcal L \subseteq \mathbb T \Sigma^+$ is {\em $\Delta$-definable} if there exists a sentence $\varphi \in \Delta$ such that $\mathcal L(\varphi) = \mathcal L$.

Let $\V = \{X_1, ..., X_m\} \subseteq V$ with $|\V| =  m$. For $\varphi \in \Rdl(\Sigma)$, let
$\exists \V. \varphi$ denote the formula $\exists X_1. \; ... \; \exists X_m. \varphi$. For a formula $\varphi \in \Rdl(\Sigma)$, let $\mathcal D(\varphi) \subseteq V_2$ denote the set of all variables $X$ for which there exist $x \in V_1$, ${\bowtie} \in \{<, \le, =, \ge, >\}$ and $c \in \mathbb N$ such that $\dd_{\leftarrow}^{\bowtie c}(X, x)$ is a subformula of $\varphi$. Let $\Rdl^{\leftarrow}(\Sigma) \subseteq \Rdl(\Sigma)$ denote the set of all formulas $\varphi$ where quantification of second-order variables is applied only to variables not in $\mathcal D(\varphi)$. We denote by $\exists \Rdl^{\leftarrow}(\Sigma) \subseteq \Rdl(\Sigma)$ the set of all sentences of the form $\exists \mathcal D(\varphi). \varphi$.

\begin{theorem}[Wilke \cite{Wil942}]
\label{Theorem:Wilke}
Let $\Sigma$ be an alphabet and $\mathcal L \subseteq \mathbb T \Sigma^+$ a timed language. Then, $\mathcal L$ is recognizable iff $\mathcal L$ is $\exists \Rdl^{\leftarrow}(\Sigma)$-definable.
\end{theorem}

\subsection{Weighted Relative Distance Logic}

\label{SSect:wRdl}

In this subsection, we consider a weighted version of relative distance logic. For untimed words, weighted MSO logic over semirings was defined in \cite{DG07}.  A weighted MSO logic over (untimed) product valuation monoids was considered in \cite{DM12}. We will use a similar approach to define the syntax and the semantics of our weighted relative distance logic. In \cite{DM12}, valuation monoids were augmented with a product operation and a unit element to define the semantics of weighted formulas. Here, we proceed in a similar way and consider timed {\em product} valuation monoids. 

A {\em timed product valuation monoid} (timed pv-monoid) $\M = (M, +, \val, \diamond, \zero, \one)$ is a timed valuation monoid $(M, +, \val, \zero)$ equipped with a multiplication $\diamond: M \times M \to M$ and a unit $\one \in M$ such that $m \diamond \one = \one \diamond m = m$ and $m \diamond \zero = \zero \diamond m = \zero$ for all $m \in M$, $\val(((\one, \one), t_1), ..., ((\one, \one), t_n)) = \one$ for all $n \ge 1$ and all $t_1, ..., t_n \in \Rp$, and $\val(((m_1, m_1'), t_1) ... ((m_n, m_n'), t_n)) = \zero$ whenever $m'_i = \zero$ for some $1 \le i \le n$. We say that $\M$ is {\em idempotent} if $+$ is idempotent.

\begin{example}
\label{Ex:TPVM}
If we augment the timed valuation monoids $\M^{\text{sum}}$, $\M^{\text{avg}}$ and $\M^{\text{disc}_{\lambda}}$ from Example \ref{Ex:TVM} with the multiplication $\diamond = +$ and the unit $\one = 0$, then we obtain the timed pv-monoids $\M_0^{\text{sum}}$, $\M_0^{\text{avg}}$ and $\M_0^{\text{disc}_{\lambda}}$. Note that these timed pv-monoids are idempotent.
\end{example}

Motivated by the examples, for the clarity of presentation, we restrict ourselves to idempotent timed pv-monoids. 

Let $\Sigma$ be an alphabet and $\M = (M, +, \val, \diamond, \zero, \one)$ a timed pv-monoid. The set $\wRdl(\Sigma, \M)$ of formulas of {\em weighted relative distance logic} over $\Sigma$ and $\M$ is defined by the grammar
$$
\varphi \; ::= \; \mathbb B. \beta \; | \; m \; | \; \varphi \vee \varphi \; | \; \varphi \wedge \varphi \; | \; \exists x. \varphi \; | \; \forall x. (\varphi, \varphi) \; | \; \exists X. \varphi
$$
where $\beta \in \Rdl^{\leftarrow}(\Sigma)$, $m \in M$, $x \in V_1$ and $X \in V_2$; the notation $\mathbb B. \beta$ indicates that here $\beta$ will be interpreted in a quantitative way.

Let $\mathbb T \Sigma^+_V$ denote the set of all pairs $(w, \sigma)$ where $w \in \mathbb T \Sigma^+$ and $\sigma$ is a $w$-assignment. For $\varphi \in \wRdl(\Sigma, \M)$, the {\em semantics} of $\varphi$ is the mapping $[\![\varphi]\!]: \mathbb T \Sigma^+_V \to M$ defined for all $(w, \sigma) \in \mathbb T \Sigma^+_V$ with $w = (a_1, t_1) ... (a_n, t_n)$ inductively on the structure of $\varphi$ as shown in Table \ref{Table:Semantics}. 
\begin{table}[t]
\begin{scriptsize}
\begin{tabular}{@{\hspace{0.42cm}}l@{\hspace{1.3cm}}l}
$
\begin{array}{rll}
\! [\![\mathbb B. \beta]\!](w, \sigma) & = & \begin{cases} \one, & \text{if } (w, \sigma) \models \beta, \\
\zero, & \text{otherwise}
\end{cases} \\
\! [\![m]\!](w, \sigma) & = & m \\
\! [\![\varphi_1 \vee \varphi_2]\!](w, \sigma) & = & [\![\varphi_1]\!](w, \sigma) + [\![\varphi_2]\!](w, \sigma) \\
\end{array}
$ 
& 
$
\begin{array}{rll}
\! [\![\varphi_2 \wedge \varphi_2]\!](w, \sigma) & = & [\![\varphi_1]\!](w, \sigma) \diamond [\![\varphi_2]\!](w, \sigma) \\
\! [\![\exists x. \varphi]\!](w, \sigma) & = & \sum\limits_{i \in \dom(w)} [\![\varphi]\!](w, \sigma[x/i]) \\
\! [\![\exists X. \varphi]\!](w, \sigma) & = & \sum\limits_{I \subseteq \dom(w)} [\![\varphi]\!](w, \sigma[X/I])
\end{array}
$
\end{tabular}
$
\! [\![\forall x. (\varphi_1, \varphi_2)]\!](w, \sigma) = \val[(([\![\varphi_1]\!](w, \sigma[x/i]), [\![\varphi_2]\!](w, \sigma[x/i])), t_i)]_{i \in \dom(w)} 
$
\end{scriptsize}
\medskip
\caption{The semantics of weighted relative distance logic}
\label{Table:Semantics}
\end{table}
Here, $x \in V_1$, $X \in V_2$, $\beta \in \Rdl^{\leftarrow}(\Sigma)$, $m \in M$ and ${\varphi, \varphi_1, \varphi_2 \in \wRdl(\Sigma, \M)}$.

\begin{remark}
In \cite{Qua10, Qua11}, Quaas introduced a weighted version of relative distance logic over a semiring $\mathbb S = (S, +, \cdot, \zero, \one)$ and a family of functions $\mathcal F \subseteq S^{\Rp}$ where elements of $S$ model discrete weights and functions $f \in \mathcal F$ model continuous weights.
If $\mathcal F$ is a one-parametric family of functions $(f_s)_{s \in S}$, then our weighted logic incorporates the logic of Quaas over $\mathbb S$ and $\mathcal F$. However, for more complicated timed valuation functions (like average and discounting) we must have formulas which combine both discrete and continuous weights. Therefore, we use the formulas $\forall x. (\varphi_1, \varphi_2)$. Our approach also extends the idea of \cite{DM12} to define the semantics of formulas with a first-order universal quantifier using the valuation function. 
\end{remark}

\begin{example}
Let $\Sigma = \{a, b\}$, $C(a), C(b) \in \mathbb R$ be the {\em continuous costs} of $a,b$ and $D(a), D(b) \in \mathbb R$ the {\em discrete costs}. Given a timed word ${w = (\gamma_1, t_1) ... (\gamma_n, t_n) \in \mathbb T \Sigma^+}$, the {\em average cost} of $w$ is defined as $A(w) = \frac{\sum_{i = 1}^n (C(\gamma_i) \cdot t_i + D(\gamma_i))}{\sum_{i=1}^n t_i}$. Let $\M_0^{\text{avg}}$ be defined as in Example \ref{Ex:TPVM}. For $U \in \{C, D\}$, let $\varphi_U(x) = (P_a(x) \wedge U(a)) \vee (P_b(x) \wedge U(b))$. Consider the 
$\wRdl(\Sigma, \M_0^{\text{avg}})$-sentence $\varphi = \forall x. (\varphi_C(x), \varphi_D(x))$. Then, for all ${w \in \mathbb T \Sigma^+}$, we have: $[\![\varphi]\!](w) = A(w)$.
\end{example}

A sentence $\varphi \in \wRdl(\Sigma, \M)$ is defined as usual as a formula without free variables. Then, for every sentence $\varphi \in \wRdl(\Sigma, \M)$, every timed word $w \in \mathbb T \Sigma^+$ and every $w$-assignment $\sigma$, the value $[\![\varphi]\!](w, \sigma)$ does not depend on $\sigma$. Hence, we can consider the semantics of $\varphi$ as a quantitative timed language $[\![\varphi]\!]: \mathbb T \Sigma^+ \to M$ over $\M$. 

Similarly to the results of \cite{DG07}, in general weighted relative distance logic and WTA are not expressively equivalent. We can show that the QTL $\mathbb L: \mathbb T \Sigma^+ \to \mathbb R \cup \{\infty\}$ with $\mathbb L(w) = |w|^2$ is not recognizable over the timed valuation monoid $\M^{\text{sum}}$. But this QTL is defined by the $\wRdl(\Sigma, \M^{\text{sum}}_0)$-sentence $\forall x. (0, \forall y. (0, 1))$. 

Nevertheless, there is a syntactically restricted fragment of weighted relative distance logic which is expressively equivalent to WTA. Let $\Sigma$ be an alphabet and $\M = (M, +, \val, \diamond, \zero, \one)$ an idempotent timed pv-monoid. A formula $\varphi \in \wRdl(\Sigma, \M)$ is called {\em almost boolean} if it is  built from boolean formulas $\mathbb B. \beta \in \Rdl^{\leftarrow}(\Sigma, \M)$ and constants $m \in M$ using disjunctions and conjunctions. We say that a formula $\varphi$ is {\em syntactically restricted} if whenever it contains a subformula $\forall x. (\varphi_1, \varphi_2)$, then $\varphi_1, \varphi_2$ are almost boolean; whenever it contains a subformula $\varphi_1 \wedge \varphi_2$, then either $\varphi_1, \varphi_2$ are almost boolean or $\varphi_1 = \mathbb B. \varphi'$ or $\varphi_2 = \mathbb B. \varphi'$ with $\varphi' \in \Rdl^{\leftarrow}(\Sigma)$; every constant $m \in M$ is in the scope of a first-order universal quantifier. Let $\Def^{\text{res}}(\Sigma, \M)$ denote the collection of all QTL $\mathbb L: \mathbb T \Sigma^+ \to M$ over $\M$ such that $\mathbb L = [\![\varphi]\!]$ for some syntactically restricted $\wRdl(\Sigma, \M)$-sentence $\varphi$. 

Our main result for weighted relative distance logic is the following theorem. 

\begin{theorem}
\label{Thm:Rec_Eq_wRdl}
Let $\Sigma$ be an alphabet and $\M$ an idempotent timed pv-monoid. Then, $\Def^{\text{\rm res}}(\Sigma, \M) = \Rec(\Sigma, \M)$.
\end{theorem}

Now we give a sketch of the proof of this theorem.  Let $\mathcal N^{\exists \Rdl^{\leftarrow}}(\Sigma, \M)$ denote the collection of all QTL ${\mathbb L: \mathbb T \Sigma^+ \to M}$ over $\M$ such that there exist an alphabet $\Gamma$, mappings $h: \Gamma \to \Sigma$, $g: \Gamma \to M \times M$ and a $\exists \Rdl^{\leftarrow} (\Gamma)$-definable timed language $\mathcal L$ such that $\mathbb L = h((\val \circ g) \cap \mathcal L)$.
For the proof of Theorem \ref{Thm:Rec_Eq_wRdl}, we establish a Nivat-like characterization of definable QTL. 

\begin{theorem}
\label{Thm:LogicNivat}
Let $\Sigma$ be an alphabet and $\M$ an idempotent timed pv-monoid. Then, $\mathcal N^{\exists \Rdl^{\leftarrow}}(\Sigma, \M) = \Def^{\text{\rm res}}(\Sigma, \M)$.
\end{theorem}

\begin{proof}[Sketch]
To show the inclusion $\subseteq$, let $\mathbb L = h((\val \circ g) \cap \mathcal L)$ where $\Gamma$, $h$, $g$ and $\mathcal L$ are as in the definition of $\mathcal N^{\exists \Rdl^{\leftarrow}}(\Sigma, \M)$. Let $\beta$ be a $\exists \Rdl^{\leftarrow}(\Sigma)$-sentence defining $\mathcal L$. We introduce a family $\V = (X_{\gamma})_{\gamma \in \Gamma}$ of second-order variables not occurring in $\beta$. We replace each predicate $P_{\gamma}(x)$ with $\gamma \in \Gamma$ occurring in $\beta$ by the formula ${P_{h(\gamma)}(x) \wedge X_{\gamma}(x)}$; so we obtain a formula $\beta' \in \exists \Rdl^{\leftarrow}(\Sigma)$. Assume that $\beta' = \exists \mathcal D(\beta''). \beta''$ with $\beta'' \in \Rdl^{\leftarrow}(\Sigma)$. We construct a formula $\text{Part} \in \Rdl^{\leftarrow}(\Sigma)$ which demands that the variables $\V$ form a partition of the domain, and  a formula $H \in  \Rdl^{\leftarrow}(\Sigma)$ which demands that, whenever a position of a word belongs to $X_{\gamma}$, then this position is labelled by $h(\gamma)$. Then, the following syntactically restricted $\wRdl(\Sigma, \M)$-sentence defines $\mathcal L$:
$$
\exists (\V \cup \mathcal D(\beta'')). \big[\mathbb B. (\beta'' \wedge \text{Part} \wedge H) \wedge \forall x. \textstyle \big(\bigvee_{\gamma \in \Gamma} \mathbb B. X_{\gamma}(x) \wedge g_1(\gamma), \bigvee_{\gamma \in \Gamma} \mathbb B. X_{\gamma}(x) \wedge g_2(\gamma) \big) \big]
$$
where, for $i \in \{1,2\}$, $g_i$ is the projection of $g$ to the $i$-th coordinate. 

To show the inclusion $\supseteq$, we introduce {\em canonical $\wRdl(\Sigma, \M)$-sentences} which are of the form $\varphi = \exists \V. \forall y. (\bigvee_{i = 1}^k \mathbb B. \beta_i \wedge m_i, \bigvee_{i = 1}^k \mathbb B. \beta_i \wedge m'_i)$ where $\V$ is a set of variables, $m_1, ..., m_k, m_1', ..., m_k' \in M$ and $\beta_1, ..., \beta_k \in \Rdl^{\leftarrow}(\Sigma)$ are such that, for every timed word $w \in \mathbb T \Sigma^+$ and every $w$-assignment $\sigma$, there exists exactly one $i \in \{1, ..., k\}$ such that $(w, \sigma) \models \beta_i$.
We can show that every syntactically-restricted sentence can be transformed into a canonical one. It remains to prove that, for a canonical sentence $\varphi$ as above, $[\![\varphi]\!] \in \mathcal N^{\exists \Rdl^{\leftarrow}}(\Sigma, \M)$. Let $M^1_{\varphi} = \{m_1, ..., m_k\}$ and $M^2_{\varphi} = \{m'_1, ..., m'_k\}$. We put $\Gamma = \Sigma  \times M_{\varphi}^1 \times M_{\varphi}^2$. Let $h: \Gamma \to \Sigma$ be the projection to the first coordinate. Let $g: \Gamma \to M \times M$ be the projection to $M_{\varphi}^1 \times M_{\varphi}^2$. Then we can construct a $\exists \Rdl^{\leftarrow}(\Gamma)$-sentence $\beta$ of the form $\exists \V. \forall y. \beta'$ such that 
$[\![\varphi]\!] = h((\val \circ g) \cap \mathcal L(\beta))$. \qed
\end{proof}

Then, our Theorem \ref{Thm:Rec_Eq_wRdl} follows from Theorem \ref{Thm:LogicNivat}, the Nivat Theorem \ref{Theorem:Nivat} and Wilke's Theorem \ref{Theorem:Wilke}.

\begin{remark}
We can also follow the approach of \cite{DG07} to prove our Theorem \ref{Thm:Rec_Eq_wRdl}. Compared to this way, our new proof technique has the following advantages. The proof idea of \cite{DG07} involves technical details like B\"uchi's encodings of assignments and a bulky logical description of accepting runs of timed automata. In our new proof, these details are taken care of by Wilke's proof for unweighted relative distance logic.
\end{remark}

Let $\Sigma$ be an alphabet, $\M^{\text{sum}}$ the timed valuation monoid as in Example \ref{Ex:TVM}(a) and $\A$ a WTA over $\Sigma$ and $\M$. As it was shown in \cite{ATP01, BFHL01, LBBF01}, $\inf\{||\A||(w) \; | \; w \in \mathbb T \Sigma^+\}$ is computable. This result and our Theorem \ref{Thm:Rec_Eq_wRdl} imply decidability results for weighted relative distance logic.
\begin{itemize}
\item 
Let $\M_0^{\text{sum}}$ be the timed pv-monoid as in Example \ref{Ex:TPVM}. It is decidable, given an alphabet $\Sigma$, a syntactically restricted sentence ${\varphi \in \wRdl(\Sigma, \M^{\text{sum}})}$ with constants from $\mathbb Q$ and a threshold $\theta \in \mathbb Q$, whether there exists $w \in \mathbb T \Sigma^+$ with $[\![\varphi]\!](w) < \theta$.
\item 
Let $\M_0^{\text{avg}}$ be the timed pv-monoid as in Example \ref{Ex:TPVM}. It is decidable, given an alphabet $\Sigma$, a syntactically restricted sentence ${\varphi \in \wRdl(\Sigma, \M^{\text{avg}})}$ with constants from $\mathbb Q$ and a threshold $\theta \in \mathbb Q$, whether there exists $w \in \mathbb T \Sigma^+$ with $\langle w \rangle > 0$ and $[\![\varphi]\!](w) < \theta$.
\end{itemize}

\section{Conclusion and Future Work}

In this paper, we proved a version of Nivat's theorem for weighted timed automata on finite words which states a connection between the quantitative and qualitative behaviors of timed automata. We also considered several applications of this theorem. Using this theorem, we studied the relations between sequential, unambiguous and non-deterministic WTA. We also introduced a weighted version of Wilke's relative distance logic and established a B\"uchi-like result for this logic, i.e., we showed the equivalence between restricted weighted relative distance logic and WTA. Using our Nivat theorem, we deduced this from Wilke's result. 

Because of space constraints, we did not present in this paper the following results. As in \cite{DM12}, for timed pv-monoid with additional properties there are larger fragments of weighted relative-distance logic which are still expressively equivalent to WTA. For the simplicity of presentation, we restricted ourselves to idempotent timed pv-monoids. However, we also obtained a more complicated result for non-idempotent timed pv-monoids. In \cite{Qua10, Qua11}, for weighted relative distance logic over non-idempotent semi\-rings, a strong restriction on the use of a first-order universal quantification was done. Surprisingly, in our result we could avoid this restriction.

Our future work concerns the following directions. The ongoing research should extend the currently obtained results to $\omega$-infinite words. This work should be further extended to the {\em multi-weighted} setting for WTA, e.g., the optimal reward-cost ratio \cite{BBL04, BBL08} or the optimal consumption of several resources where some resources must be restricted \cite{LR05}. A logical characterization of untimed multi-weighted automata was given in \cite{DP13}. It could be also interesting to compare for the weighted and unweighted cases the complexity of translations between logic and automata. We believe that our Nivat theorem will be helpful for this.

\end{document}